\documentclass[11pt]{article}
\usepackage[utf8]{inputenc}
\usepackage{amsmath}
\usepackage{amsthm}
\usepackage{graphicx}
\usepackage{epstopdf}
\usepackage{epsfig}
\usepackage{array}
\usepackage{amssymb}
\usepackage{eso-pic}
\usepackage{caption}
\usepackage{cite}
\usepackage{subcaption}
\usepackage{verbatim}
\usepackage{url}
\usepackage{multirow}
\usepackage{color, soul}
\usepackage{bm}
\usepackage{syntonly}
\usepackage{algorithm} 
\usepackage{algpseudocode}

\usepackage{graphicx}
\usepackage[left=1.0in,right=1.0in]{geometry}
\usepackage{CJKutf8}
\usepackage{mathptmx}
\usepackage{setspace}
\newtheorem{theorem}{Theorem}

\def\scriptO{{{\it O}\kern -.42em {\it `}\kern + .20em}}

\author{Jared Cook$^1$, Ralph C. Smith$^1$, Camila Ramirez$^2$, Nageswara S. V. Rao$^2$}
\date{$^1$Department of Mathematics, North Carolina State University, Raleigh, NC \\ $^2$Center for Engineering Science Advanced Research, ORNL , Oak Ridge, TN}
\title{Particle Filtering Convergence Results for Radiation Source Detection}

\begin{document}

\maketitle

\section*{Abstract}
Recent research has shown a weak convergence -- convergence in distribution -- of particle filtering methods under certain assumptions. However, some applications of particle filtering methods, such as radiation source localization problems, can be shown to have an extended convergence in the following sense. Using the assumptions of statistically independent measurements and a measurement process that does not depend on the state space, we prove the convergence of the posterior and its approximation by the particle filtering algorithm to the true Dirac distribution characterizing the source location. We design a Sampling Importance Resampling (SIR) filter to detect and locate a radiation source using a network of sensors. To numerically assess the effectiveness of this particle filter we employ it to solve a source localization problem using both experimental open field data sets from Intelligent Radiation Sensor Systems (IRSS) tests and numerically simulated urban domain data using a ray-tracing model. We then prove the theoretical convergence of the estimate to the state posterior probability density function. 

\section{Introduction}\label{intro}

With a focus of security applications, we consider the problem of estimating the location and intensity of an unknown radioactive source given count rate measurements provided by detectors at known locations. These networks of radiation detectors are increasingly being deployed to detect and localize sources of radiation within urban areas, for portal monitoring, at special events, and at boarder crossings. Because of this, an accurate model for processing the detector responses is required to detect and locate a given source quickly and accurately. Current research has focused on fusing the sensor data to determine the source location while attempting to decrease the number of sensors, increase the accuracy of the location, and increase the efficiency of the algorithm \cite{sensornet, Rao, C3, Razvan, refRao15}.

Particle filtering is a sequential Monte Carlo method for solving signal processing problems, which involves the estimation of posterior distributions. An essential component comprises the approximation of the internal states in dynamical systems when partial observations are made and random perturbations are present. To implement the particle filtering algorithm, we employ point masses, or “particles”, and their associated importance weights to estimate the posterior probability density function of the state, based on all available information \cite{C5, C1}. 

We use sequential importance sampling and Bayesian inference to solve a radiation source localization problem. We implement a Sampling Importance Resampling (SIR) particle filter designed to detect and locate a radiation source using a network of detectors. The Domestic Nuclear Detection Office's (DNDO) Intelligence Radiation Sensors Systems (IRSS) program supported the development of commercial radiation sensor networks for use in detection, localization and identification of primarily low-level radiation sources. Through this program, a range of indoor and outdoor tests were conducted using stationary and moving sources of varying strengths and types, different background profiles, and varying detector layouts. From the information gathered in these tests, canonical datasets were packaged for public release \cite{data}. We numerically assess the effectiveness of this particle filter by applying it to a subset of these datasets. 

In prior work, we employed a spatially two-dimensional ray-tracing algorithm to model detector responses to a radiation source in a simulated city block \cite{Razvan}. We used the open source code \texttt{gefry} to solve for the detector responses within this test environment \cite{gefry}. Whereas this algorithm is based on Boltzmann transport theory, scattering is neglected to significantly improve the efficiency. We employ this ray-tracing model to test the performance of the particle filtering algorithm on a more complex simulated urban source localization problem. Additionally, we employ the ray-tracing model and the particle filtering algorithm to implement a simple moving detector strategy for which we show improved source localization results, as expected. 

Lastly, we prove the theoretical convergence of the particle filtering estimate to the state posterior probability density function (PDF). There have been many developments in showing weak convergence in mean squared error of the empirical particle filtering distribution toward their true values \cite{convergence}. For our radiation source detection problem, we have statistically independent measurements and a measurement process that does not depend on the state space. These properties allow us to prove the particle filter applied to the radiation source detection problem converges to the underlying Dirac density centered at the true source characteristics. The assumption of the underlying distribution being a Dirac distribution is motivated by the large size of the domain, which allows us to treat the source as a point source.

This paper is organized as follows. In Section~\ref{opt_filter}, we begin by describing the model and the optimal filtering problem in a generic framework. Section~\ref{pf_section} discusses the particle filtering algorithm used in our analysis and in Section~\ref{rsl_section}, we formulate  the radiation source localization problem. In Sections~\ref{num_res}  and \ref{raytrace_res}, we provide numerical source localization results for our particle filtering algorithm using the test data from by the IRSS tests and the simulated data using a ray-tracing model. In Section~\ref{conv_analysis}, we present the convergence results for our particle filtering algorithm. We note that this section is further subdivided into proving that the the empirical distribution converges to its true underlying distribution and proving that the particles converge to the Dirac distribution centered at the true source location. Finally, we summarize in Section~\ref{conclusions} conclusions and future research. 

\section{Optimal Filtering}\label{opt_filter}

To define the optimal filtering problem, consider the unobservable state sequence within the state space $\Omega$ and the sequence of partial or noisy measurements 
\begin{equation} 
\{x_k\in \Omega = \mathbb{R}^{n_x}, k\in \mathbb{N}_{\geq 0}\},  \label{eq-seq1}
\end{equation}
\begin{equation}
\{y_k\in \mathbb{R}^{n_y}, k \in \mathbb{N} \}, \label{eq-seq2}
\end{equation} 
which are Markov processes represented by the equations
\begin{equation}
x_{k}=f_k(x_{k-1},u_{k-1}),\\  \label{eq-state} 
\end{equation}
\begin{equation}
y_{k}=h_k(x_{k},v_k). \label{eq-obs}
\end{equation}

\noindent Here 
\begin{align*}
&f_k(\cdot ):\mathbb{R}^{n_x}\times \mathbb{R}^{n_u}\rightarrow \mathbb{R}^{n_x}, \quad \text{and} \\
&h_k(\cdot ):\mathbb{R}^{n_x}\times \mathbb{R}^{n_v}\rightarrow \mathbb{R}^{n_y}
\end{align*}
are the known, time-dependent, and possibly nonlinear system transition and measurement functions. We note that $u_{k-1}$ and $v_k$ are independent and identically distributed (iid) state and measurement noise vectors of size $n_u$ and $n_v$ respectively and $k$ is the time index, $k=1,\hdots ,T$. 

The optimal filtering problem consists of estimating the state sequence (\ref{eq-seq1}) recursively using the information provided by the observation process (\ref{eq-seq2}). In a Bayesian setting, this process can be formalized as the computation of the filtering distribution $\pi(x_k|y_{1:k})$, where $y_{1:k}$ denotes the measurements $y_{1},\hdots, y_{k}$. Here, we assume that the prior PDF $\pi(x_0)$ of the state vector is available. When this is the case, the computation of the filtering distribution can be performed recursively in two steps: prediction and update. 

In the prediction step, we assume that the filtering distribution from the previous recursion $\pi(x_{k-1}|y_{1:k-1})=\pi_{k-1}$ is available, and we compute the posterior distribution

\begin{equation}\label{eq-predict1}
\pi(x_{k}|y_{1:k-1})= \int_{\mathbb{R}^{n_x}} \pi(x_{k}|x_{k-1})\pi(x_{k-1}|y_{1:k-1})dx_{k-1}.
\end{equation}
The probabilistic model of the state evolution $\pi(x_{k}|x_{k-1})$ is defined by (\ref{eq-state}) and the known state noise statistic $u_{k-1}$. Next, we assume that at time $k$, we obtain a new measurement $y_k$ and we employ Bayes' relation 
 to compute the updated PDF   
\begin{equation} \label{eq-update1}
\pi(x_k|y_{1:k}) = \frac{\pi(y_k|x_k) \pi(x_k|y_{1:k-1})}{\pi(y_k|y_{1:k-1})} \equiv \pi_k. 
\end{equation}
Here, the normalizing constant is $\pi(y_k|y_{1:k-1}) =  \int_{\mathbb{R}^{n_x}} \pi(y_k|x_k) \pi(x_k|y_{1:k-1}) dx_{k}$, where the likelihood function $\pi(y_k|x_k)$ is defined by (\ref{eq-obs}) and the known measurement noise statistic $v_k$ \cite{convergence}.

 To illustrate how the state evolution model and likelihood function are defined, consider the example of Kalman filtering where we assume that the state (\ref{eq-state}) and measurement (\ref{eq-obs}) equations are linear 
\begin{align*}
&x_{k}=A x_{k-1} + Bu_{k-1}, \\
&y_{k}=C x_{k} + D v_k.
\end{align*}
Given this, we define the probabilistic state evolution and the likelihood as
\begin{align*}
&\pi(x_{k}|x_{k-1}) \sim \mathit{N}(Ax_{k-1}, BB^T), \\
&\pi(y_k|x_k) \sim \mathit{N}(Cx_k, DD^T),
\end{align*} 
where A, B, C, and D are matrices of appropriate size. We explicitly define the state evolution model and likelihood function for our specific problem in Section~\ref{rsl_section}.

The recursive propagation of the posterior density in (\ref{eq-predict1}) and (\ref{eq-update1}) is computationally intractable for large state spaces $n_x$ due to the integration required to compute $\pi (y_k|y_{1:k-1})$ and $\pi(x_k|y_{1:k-1})$, so we must resort to approximation methods such as Monte Carlo approximation. 

\section{Particle Filtering}\label{pf_section}
In particle filtering, one uses particles and corresponding weights whose empirical measure approximates the target posterior distribution $\pi(x_{0:k}|y_{1:k})$. Here, we summarize the Sampling Importance Resampling (SIR), or bootstrap filter, algorithm, which uses particle and weight pairs $\{p_k^i,w_k^i\}_{i=0}^N$ to approximate the target distribution as

\begin{equation*}
    \pi(x_{0:k}|y_{1:k})\approx \sum_{i=1}^N w_k^i \delta(x_k-p_k^i).
\end{equation*} 
Here, $\delta(x_k-p_k^i)$ is the Dirac distribution centered at $p_k^i$.  We compute the weights using importance sampling and normalize them such that $\sum_{i=1}^N w_k^i = 1$; e.g., see \cite{C5}. 

To illustrate importance sampling, suppose we want to approximate the distribution $\pi(x)$ with discrete random measures $p^i$. If we were able to generate points from $\pi(x)$, each would be assigned an equal weight of $1/N$. When direct sampling is not possible, we employ importance sampling by sampling from the importance density $p^i \sim \pi^N(x)$ instead, where $\pi^N$ is chosen to ensure efficient sampling. We then assign the weights as

\begin{equation*}
    w^i = \frac{\pi(p^i)}{\pi^N(p^i)}
\end{equation*} 
and normalize the weights so they sum to one.

Assume now that the posterior distribution $\pi(x_{0:k-1}|y_{1:k-1})$ is approximated by the discrete random measure $\{p_{k-1}^i,w_{k-1}^i\}^N_{i=1}$ such that the particles are distributed according to $\pi(x_{0:k-1}|y_{1:k-1})$. If we employ an importance function that can be factored as 

\begin{equation*}
    \pi^N(x_{0:k}|y_{1:k}) = \pi^N(x_k|x_{0:k-1},y_{1:k})\pi^N(x_{0:k-1}|y_{1:k}),
\end{equation*}
and if 

\begin{equation*}
    w_{k-1}^i\propto \frac{\pi(p_{0:k-1}^i|y_{1:k})}{\pi^N(p_{0:k-1}^i|y_{1:k})},
\end{equation*}
then we can augment the particles $p_{0:k-1}^i$ with $p_k^i \sim \pi^N(x_k|p_{0:k-1}^i,y_{1:k})$. Employing Bayes' relation, the associated weights are then recursively defined as 
\begin{equation}\label{eq-theoreticalweights}
    w_k^i \propto \frac{\pi(y_k|p_k^i)\pi(p_k^i|p_{k-1}^i)}{\pi^N(p_k^i|p_{0:k-1}^i,y_{1:k})}w_{k-1}^i.
\end{equation}
We generally want the importance distribution $\pi^N(x_{0:k}|y_{1:k})$ to be as close to $\pi(x_{0:k}|y_{1:k})$ as possible. We also employ the likelihood $\pi(y_k|p_k^i)$ and prior $\pi(p_k^i|p_{k-1}^i)$ in this computation and we note that these are often employed as the importance distribution. 

In this way, the recursive SIR particle filtering algorithm is divided into three steps. The prediction step is performed by sampling particles via importance sampling

\begin{equation} \label{eq-predict}
\{p_k^i\}_{i=0}^N \sim \pi^N(x_k|x_{k-1}) \equiv \pi_{k-1}^N. 
\end{equation} 
These particles are used to construct the empirical distribution $\pi^N(x_{k}|y_{1:k-1})$, which is an approximation of $\pi(x_{k}|y_{1:k-1})$. The updating step provides the empirical posterior distribution 

\begin{equation} \label{eq-update}
\tilde{\pi}^N(x_k|y_{1:k}) = \sum_{i=1}^N w_k^i \delta(x_k-p_k^i) \equiv \tilde{\pi}_k^N,
\end{equation}
where the importance weights are calculated according to (\ref{eq-theoreticalweights}) and normalized so that $\sum_{i=0}^N w_k^i = 1$.

\par The third step is introduced to address a major problem with particle filters, termed the degeneracy problem or sampling impoverishment.  This problem occurs when after several iterations only a few of the particles have significant weights and all the other particles have negligible weights. A common way to address this issue is by resampling, where a new set of particles is drawn from the discrete approximation to the filtering distribution, $p_k^i \sim \tilde{\pi}^N(x_k|y_{1:k})$. After this third resampling step, the particle weights are all set to $1/N$ and the final empirical distribution approximating the posterior is given by 

\begin{equation}\label{eq-posteriorEst}
\pi^N(x_k|y_{1:k}) = \frac{1}{N}\sum_{i=0}^N \delta (x_k-p_k^i) \equiv \pi_k^N.
\end{equation}

\section{Radiation Source Localization}\label{rsl_section}
Unlike Kalman Filtering methods, which rely on the assumptions of linearity and Gaussian distributions, particle filtering is suitable for radiation source detection and localization, since they are nonlinear problems with high variance Poisson measurements. In this framework, the particles represent potential sources and they dynamically evolve using information from the detector measurements. The particles that are more likely to be sources are retained following the resampling step and cluster around the true source. 

\par Here, we describe a Sampling Importance Resampling (SIR) particle filter designed to locate a fixed radiation source comprised of a location vector and radiation strength scalar. We use a fixed number of detectors, each having a set number of measurements. The particles for this problem, $p_k^i = (s_k^i,t_k^i,I_k^i)\in \Omega$, consisting of the two-dimensional position $(s_k^i,t_k^i)$ and radiation source strength $I_k^i$, are each drawn from a uniform distribution over their respective domains $\Omega = [s_{min},s_{max}]\times[t_{min},t_{max}]\times[I_{min},I_{max}]$.

\par The prediction step takes $p_k^i = p_{k-1}^i$, since the importance sampling function is independent of the measurement $y_k$. Therefore, the state space is explored without any knowledge of the observations and the state process $f_k(\cdot)$ is taken to be the identity matrix of appropriate size. If prior information was known about the source location or trajectory, we would use this information in our formulation of the state process, but we assume here that no information is known \textit{a priori}. 

\par The weights corresponding to each particle represent the likelihood that the particle accurately approximates the true source characteristics based on detector measurements that are Poisson distributed. We have in (\ref{eq-theoreticalweights}) that the weights are proportional to the likelihood $w_k^i \propto \pi(y_k|p_k^i)$ so we approximate the weights by setting them equal to the log of the likelihood. We take the likelihood to be a Poisson distribution with a mean value provided by a radiation detector model which approximates the detector responses to a source represented by a particle $p_k^i$. We compare the detector model responses with the observed detector responses via the Poisson distribution. Therefore, we employ the weights 

\begin{equation} \label{eq-weights}
\begin{aligned}
w_k^i&=\log\left(\pi({y}_k|p_k^i)\right) = \log\left(P(y_k; u^K_j(p_k^i))\right) = \log\left(\prod_{j=1}^d P(y_k^j; u^K_j(p_k^i))\right) \\
&= \sum_{j=1}^d \log\left(P(y_k^j; u^K_j(p_k^i))\right)= \sum_{j=1}^d \log\left(\frac{u^K_j(p_k^i)^{y_k^j}e^{-u^K_j(p_k^i)}}{y_k^j!}\right),
\end{aligned}
\end{equation} 
since the measurements $y_k=[y_k^1,\hdots, y_k^d]$ are independent. We take measurement process $h_k$ to be the Poisson distribution $P(\ \cdot \ ;  \ \cdot \ )$ with parameter provided by the detector model response $u^K(\cdot)$. Since the IRSS experiments were performed in an open field environment, we employ the $1/\text{distance}^2$ -- i.e., quadratic attenuation (QA) -- detector model
\begin{equation*}
u_j^{QA}(p_k^i) = \frac{I_k^i  \cdot \epsilon_j \cdot A_j \cdot \Delta t_j}{4\pi ||r_j-r_k^i||_2^2} = \frac{I_k^i  \cdot \epsilon_j \cdot A_j \cdot \Delta t_j}{4\pi [(s_j-s_k^i)^2+(t_j-t_k^i)^2]}
\end{equation*} 
when using the experimental data employed in Section~\ref{num_res}. Here, we denote the detector locations as $r_j = (s_j,t_j)$ and $A_j$, $\epsilon_j$, and $\Delta t_j$ are the surface area, intrinsic efficiency, and dwell time of the $j$th detector. 

We also employ a ray-tracing (RT) model to determine the expected detector count rates when we use simulated data. The full derivation of this model from the Boltzmann transport equation is provided in \cite{Bkgrd}. For each detector, we construct a ray from the source or particle location to the detector location. We compute the number of intersections with the $\mathit{N}$ buildings along the path of this ray, and the length $\ell_h$ of the ray segment through each of the $h=1,\hdots,\mathit{N}$ buildings. The buildings in the domain are assumed to be homogeneous, each having a mean free path $\lambda^h, \ h=1,\hdots,\mathit{N}$ \cite{Razvan}. These assumptions yield
\begin{equation} \label{eq-numerMod}
u_j^{RT}(p_k^i)=I_k^i \frac{\Delta t_j\cdot \epsilon_j\cdot A_j}{4\pi ||r_k^i-r_j||_2^2}\text{exp}\Bigg(-\sum_{h=1}^{\mathit{N}}\frac{\ell_h}{\lambda^h}\Bigg). 
\end{equation} 
We employ the Python code \texttt{gefry} \cite{gefry} to implement this numerical model, which is a significant simplification of the original problem derived from Boltzmann transport theory.

The model $u_j^{K}(\cdot)$, $K\in \{QA, \ RT \}$ approximates the expected count rate due to the source of the $j$th detector. The summation in (\ref{eq-weights}) is taken over the total number of detectors $d$. We note that the measurements $y_k$, $k=0,\hdots,T$, in (\ref{eq-weights}) are statistically independent of each other. We also note that this approach of computing the weights is similar to assigning these weights the values of a Gaussian likelihood 
\begin{equation}\label{eq-Likelihood}
\pi(y|q) = \frac{1}{(2\pi \sigma^2)^{d/2}}e^{-SS(q)/2\sigma^2}
\end{equation}
where 
\begin{equation}\label{eq-sos}
SS(q) = \sum^d_{j=1}[y_j-u_j^K(q)]^2, \ K\in \{QA, \ RT \}
\end{equation}
when the detectors record counts greater than approximately 30. However, the experimental data we employ in Section~\ref{num_res} was performed using low-level sources and the detectors provided counts each second. Therefore, detectors far from this low level source do not have count rates consistently greater than 30 for one second of dwell time. In Section~\ref{raytrace_res}, we simulate sources with larger intensities and we note that employing a Gaussian likelihood in that application yields similar localization results to when a Poisson likelihood is employed. 

To address the issue of sampling impoverishment, at every iteration we resample 60\% of the particles. To perform this resampling strategy, we order the particles by their weights and we retain the 40\% highest weighted particles whereas the bottom 60\% are replaced by particles drawn from a uniform distribution. This simple strategy is used to decrease the complexity of the convergence analysis, but does not perform as well numerically as other strategies, such as Multinomial or Systematic resampling \cite{Camila}. We leave the implementation of these resampling methods in the context of this application as future work. 

Lastly, we provide the SIR particle filtering algorithm for the radiation source localization problem in Algorithm~\ref{alg-pf}. We note that we take the prior distribution to be uniform over the state space $\pi (x_0) = \mathit{U} (\Omega)$ and we take the importance distribution to be equal to the prior $\pi^N( x_k | x_{k-1}) = \pi(x_0)=\mathit{U}(\Omega)$ to begin our analysis. Therefore, the particles are repositioned in the resampling step with no knowledge of the posterior estimate. 

\begin{algorithm}[H]
   Sample particles $p_0^i, i=1,\hdots,N$ from the prior distribution $p_0^i \sim \pi(x_0)$. \\
    For $k=1,\hdots,M$
 \begin{algorithmic}[1]
    	\item Take measurement $y_k$.
	\item Compute the weights $w_k^i, i=1,\hdots,N$ according to (\ref{eq-weights}).
	\item Normalize weights so that $\sum_{i=0}^N w_k^i = 1$. 
	\item Reorder particles by their weights and resample particles $p_k^i, i=1,\hdots \text{floor}(0.6\times N)$ from the importance distribution $p_k^i \sim \pi^N(x_k | x_{k-1})$.
	\item Compute posterior estimate according to (\ref{eq-posteriorEst})
    \end{algorithmic} 
    \caption{SIR Particle Filtering Algorithm}
    \label{alg-pf}
\end{algorithm}

\section{Numerical Results for Experimental Data}\label{num_res}
The Intelligence Radiation Sensor Systems (IRSS) program was initiated in 2009 by the Department of Homeland Security's Domestic Nuclear Detection Office (DNDO). The objective was to demonstrate that a system of networked detectors could outperform, in both efficiency and accuracy, an individual detector in the detection, localization, and identification of a radiation threat. The lack of experimental datasets for testing this hypothesis was addressed by performing a series of IRSS tests and creating canonical datasets. 

Three independent companies developed networked systems for the IRSS characterization tests and each developed their own detection, localization, and identification capabilities. The IRSS tests were conducted with multiple experiments performed for each of several indoor and outdoor configurations with multiple source strengths and types, different background profiles, and various types of source and detector movements. The majority of the tests were carried out with cesium-137 and cobalt-57 isotopes as sources; however, we only employ the experiments in which cesium was used as the source. The tests were conducted at the Savannah River National Laboratory (SRNL) and the indoor tests were performed in the Low Scatter Irradiator (LSI) facility at SRNL. 

The spectral counts -- counts of gamma particles with different energies -- collected by the detection devices were mapped into 21 spectral bins. This mapping was based on the selected set of isotopes and on the energy resolution of the detectors employed \cite{data}. The tests we employ in this paper to assess the accuracy of the SIR particle filter are labeled LSI~A04, LSI~C01, LSI~C02, LSI~C03 and LSI~C04 and the experimental layout of each is described in Table~\ref{tab-datasets}. For each of these experiments, the detectors and a source are stationary in the 10~m by 10~m indoor domain. The detectors are NaI 2 inch by 2 inch detectors randomly scattered around the source.

\begin{table}[!t]
    \centering
    \caption{Description of the experimental test sets.}
    \begin{tabular}{|c|c|c|c|}
    \hline \hline
        Tests & $I_{source} \ (\mu\text{Ci})$ & $I_{source} \ (\text{Bq})$ & $(x_{source},y_{source})$ in (cm)  \\
        \hline
        LSI~A04 & 35 & $1.295\times10^6$ & (0, \ 0) \\
        LSI~C01 & 7.6 &  $2.812\times10^5$ & (0,\ 0) \\
        LSI~C02 & 7.6 & $2.812\times10^5$  & (71, \ 71) \\
        LSI~C03  & 7.6 & $2.812\times10^5$  & (141, \ 141) \\
        LSI~C04  & 7.6 & $2.812\times10^5$  & (283,\ 283) \\
    \hline \hline
    \end{tabular}
    \label{tab-datasets}
\end{table}

The datasets are packaged with a MATLAB dashboard that provides an overview of the corresponding dataset. For each scenario, the dashboard plots the detector layout, the measurements of the nearest and furthest detectors over time, the counts from a spectral bin of the nearest detector, and the spectra of all detectors. Figures \ref{fig-d1} and \ref{fig-d2} depict the standard dashboard plots for the LSI~A-04 test. We see from the detector layout in Figure~\ref{fig-d1} that detector 14 is the farthest from the source and detector 6 is the nearest. We note that the LSI A04 experiment was conducted with 21 detectors whereas the LSI C experiments employed 22 detectors. The layout of the 22 detectors in the LSI C test cases is similar to the layout in the LSI A04 experiment. The 12th bin includes counts of particles with energies between $621 \text{ KeV}$ and $704 \text{ KeV}$ which corresponds to the 662 KeV cesium-137 photopeak. Figure~\ref{fig-d1}(d) depicts the time series measurements from the 12th spectral bin of the nearest detector to the source.

\begin{figure}
\centering
\includegraphics[height=120mm, width=150mm]{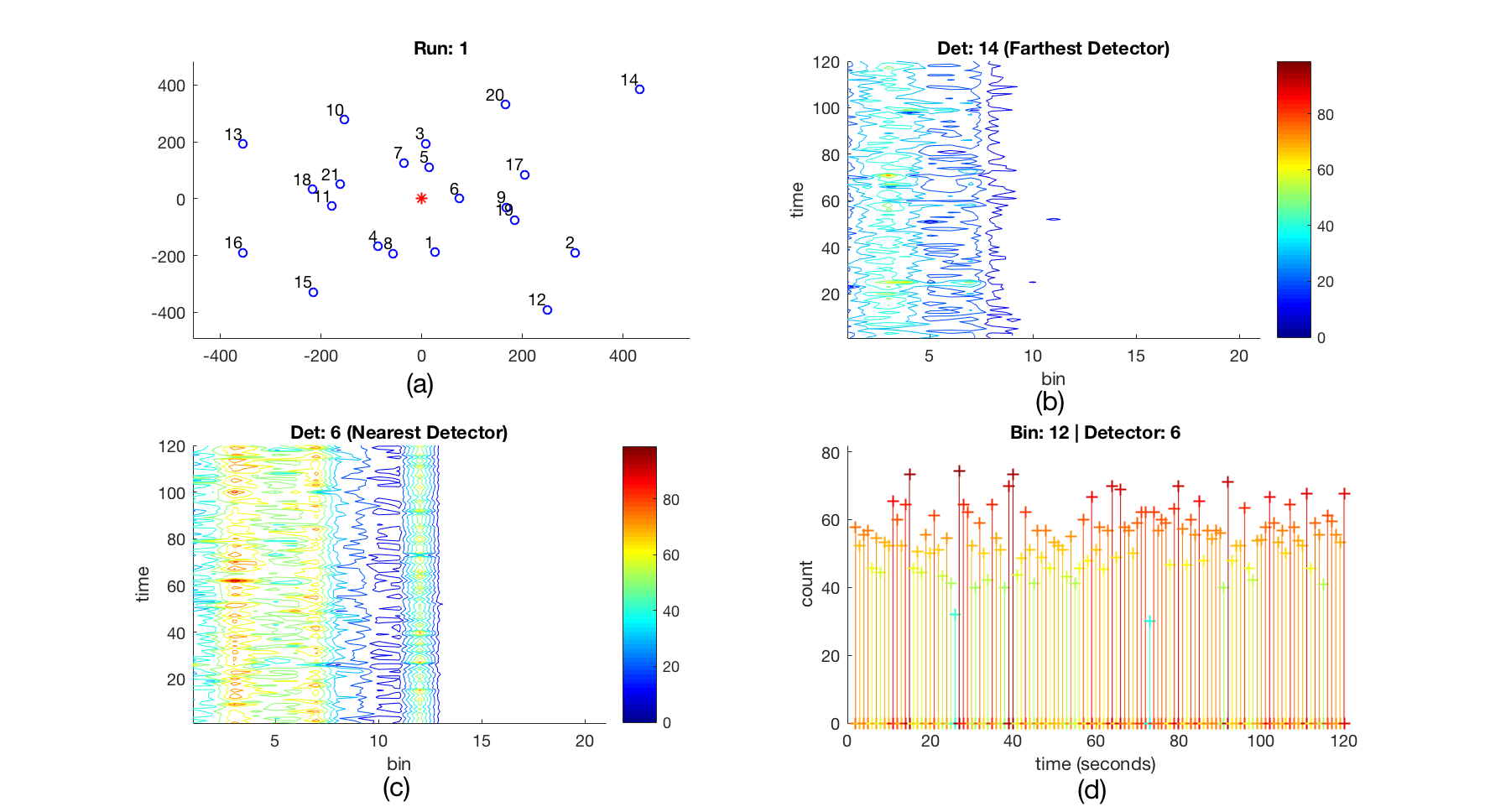}
\caption{(a) Experimental layout of the detectors, binned spectra of the (b) farthest detector from source location and (c)  nearest detector to source, and (d) time-series counts from the 12th spectral bin, associated with the cesium-137 photopeak, of the nearest detector to the source.}
\label{fig-d1}
\end{figure} 

\begin{figure}
\centering
\includegraphics[height=150mm, width=170mm]{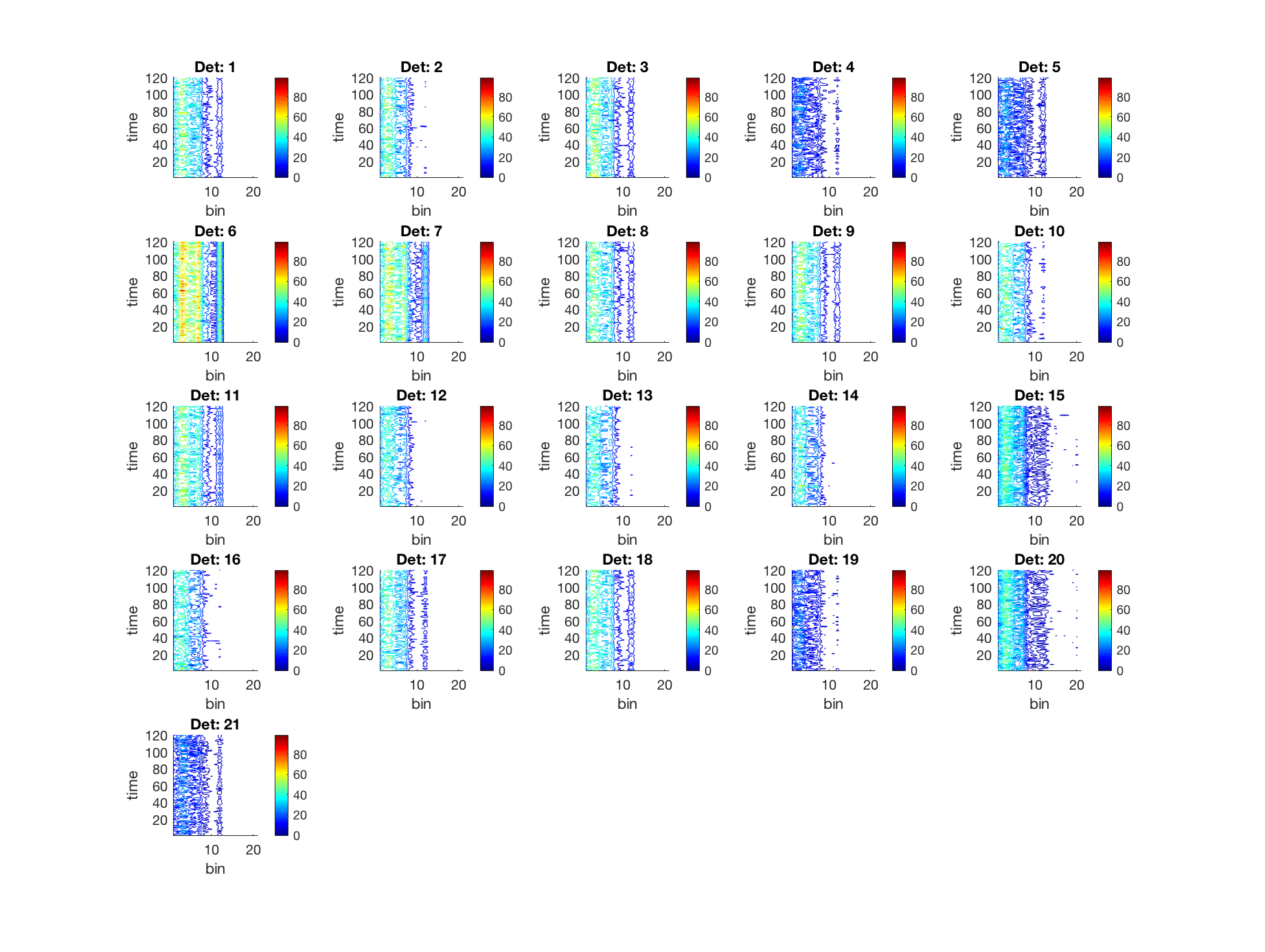}
\caption{Experimental time series binned spectral counts from each detector.}
\label{fig-d2}
\end{figure}

In addition, background measurements were taken by the detectors for the experimental setup of LSI~A04, mapped similarly into 21 spectral bins, and compiled in the LSI~A-Background dataset. The detectors in the LSI~C datasets are positioned slightly differently and include an additional detector. Furthermore, the background measurements were taken every second for only 60 seconds, whereas for the experiments where a source is present, measurements were taken every second for 120 seconds. 

To incorporate these background measurements, we find for each detector in the LSI~C datasets the closest corresponding detector in the LSI~A-Background dataset.  We simulate detector background measurements by drawing from a Poisson distribution with a mean given by the mean experimental background measurement.  Then, for each detector in each of these datasets, we subtract the Poisson distributed background sample from the detector count rate to provide us with the counts from the radiation source. We account for negative values after subtracting the Poisson simulated background by setting those values to zero. 

The first dataset, LSI~A04, was an indoor test with one stationary source and 21 stationary detectors.  Nine repetitive experiments were executed, with each run lasting approximately $T=120$ seconds and counts taken every second, but we employ only the first experiment in each case in testing the particle filter. We draw the particles from the uniform state space $\mathit{U}[\Omega]=\mathit{U}[(-5, \ 5 \text{m})\times (-5, \ 5 \text{m})\times (1e4 \text{ Bq}, \ 1e7 \text{ Bq}])]$ and employ $N=1000$ particles. For our particle filter, we choose to resample $f\times N$ of the particles at each time step $k$. We choose $f=0.6$ and we resample the $f\times N$ particles that have the lowest corresponding weights from the same uniform distribution used originally to generate the particles. The final scattering of the particles is shown in Figure~\ref{fig-fig1}(a). We note that the particles converge to what looks like a discrete approximation of a normal distribution centered near the true source location and with little variance. 

The other datasets we use, LSI~C01, LSI~C02, and LSI~C03, are also indoor tests with a stationary source and 22 stationary detectors. Again, the source, cesium-137 with radiation strength of 7.6 $\mu$Ci, was stationary within the 10 meter by 10 meter domain. As seen in Figure~\ref{fig-fig1}(b) and compiled in Table~\ref{tab-datasets}, the source was placed in the center of the domain in LSI~C01, but the sources in the LSI~C02, C03, and C04 experiments were placed in the northeast area of the domain, increasingly farther from the center. 

We note that the particle filtering algorithm fails to locate the source in this case, although there is a small cluster of particles at the source location. This is due in part to the low source strength, which is nearly an order of magnitude lower than in the LSI~A04 experiment, but may also be due to the simplified resampling scheme we employ to prove analytic convergence. Numerical convergence was shown for this same problem when more sophisticated resampling schemes were utilized \cite{Camila}. However, those results employed an additional scaling factor that artificially inflated the detector counts. We note that we obtain convergence of this algorithm when a similar scaling factor is employed in our algorithm, however we instead constrain our prior distribution to investigate whether limited prior information allows for convergence of this algorithm with the low-level sources employed to obtain this experimental data. 

Source multilateration requires that the source be within the convex hull of the detectors \cite{Cook_thesis}, which is the case for all of the experimental test cases considered in this paper. To test the effectiveness of this simple particle filtering algorithm, we draw the particles from a uniform distribution over the convex hull of the detectors. To do this, we construct the normal vectors between the detectors and use them to test whether a randomly drawn particle is within the convex hull. If it is, then a weight is calculated for it within the particle filtering algorithm, but if not, that point is discarded and another one drawn. 

We plot the results for this analysis using datasets LSI~C01 and C03 in Figure~\ref{fig-conv}. We note that there are clusters of particles, representing regions of high probability, at the source location. However, the right edge of the convex hull domain has a large cluster of particles as well. Again, this is likely due to the low levels of radiation in the experiments produced by the sources and the fact that we are attempting to infer the source location as well as the source intensity. Similar detector readings can be caused by sources of vastly different strengths depending on the source locations.  To test if the algorithm results converge, we instead sample from a Poisson distribution with mean given by the mean response from the experimental detector readings to simulate more detector measurements. 

By employing both the uniform distribution over the convex hull as the sampling distribution and the Poisson distribution to simulate more data, we observe particle convergence to the source location. We plot these results for the datasets LSI~C02 and C04 in Figure~\ref{fig-convPois} employing 1000 seconds of simulated data. We observe that the particles have converged after approximately 300 seconds. Here, we are able to localize the source to within approximately 1.5 meters given less than 5 minutes of measurement time and the prior information that the source is within the convex hull of the detector network. 

\begin{figure}[H]
    \centering
    \begin{subfigure}[b]{0.49\textwidth}
        \includegraphics[width=\textwidth]{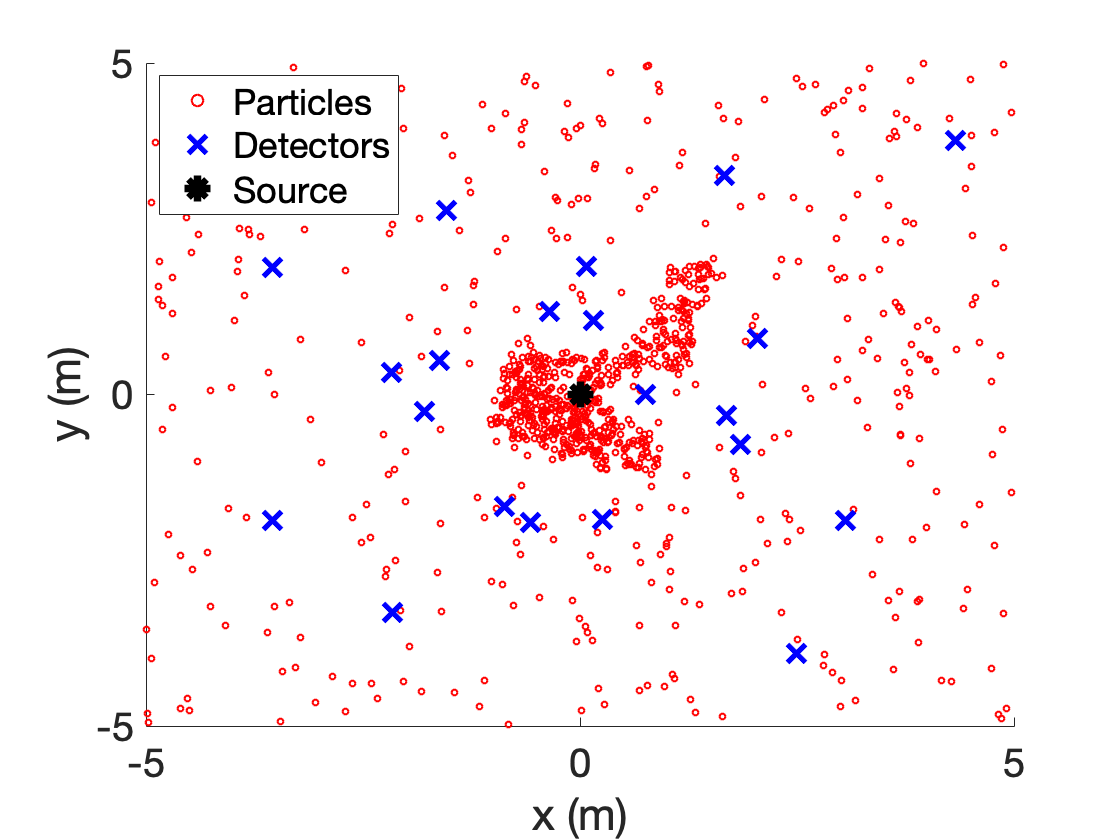}
    \end{subfigure}
    \begin{subfigure}[b]{0.49\textwidth}
        \includegraphics[width=\textwidth]{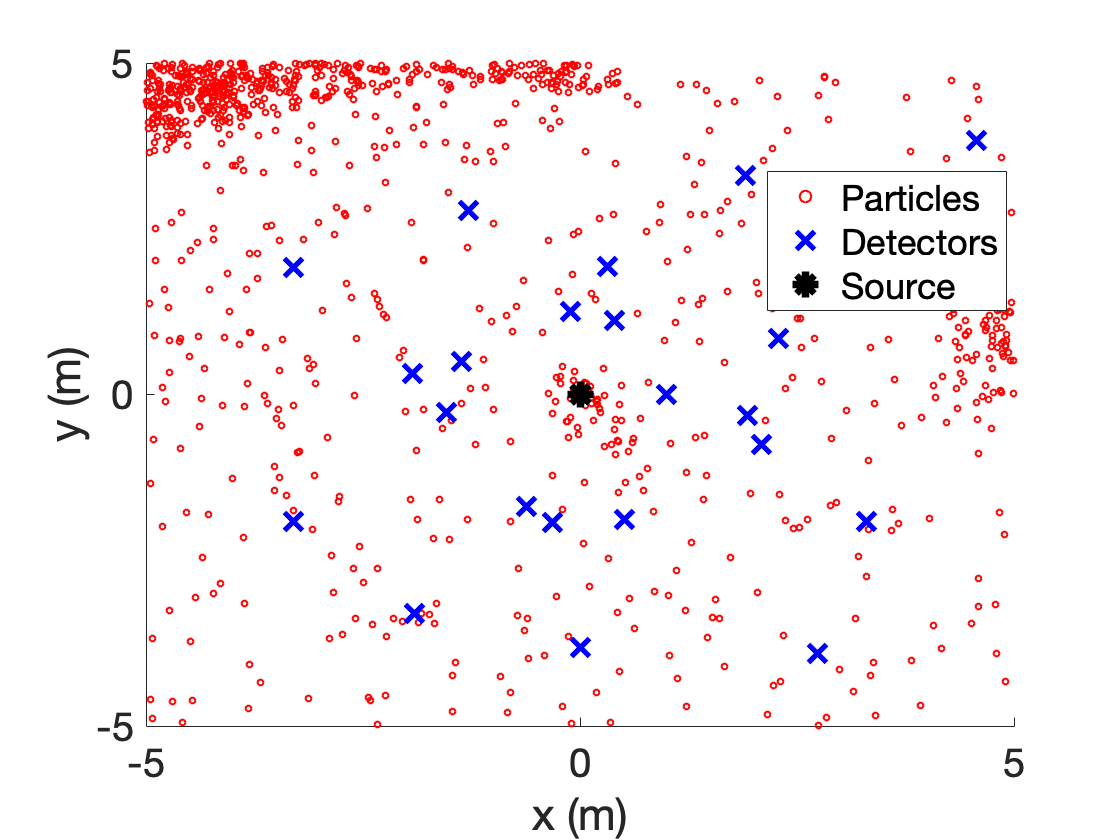}
    \end{subfigure}    
    \caption{Particle filtering results for the (a) LSI~A04 and (b) LSI~C01 experimental layouts.}
    \label{fig-fig1}
\end{figure}

\begin{figure}[H]
    \centering
    \begin{subfigure}[b]{0.49\textwidth}
        \includegraphics[width=\textwidth]{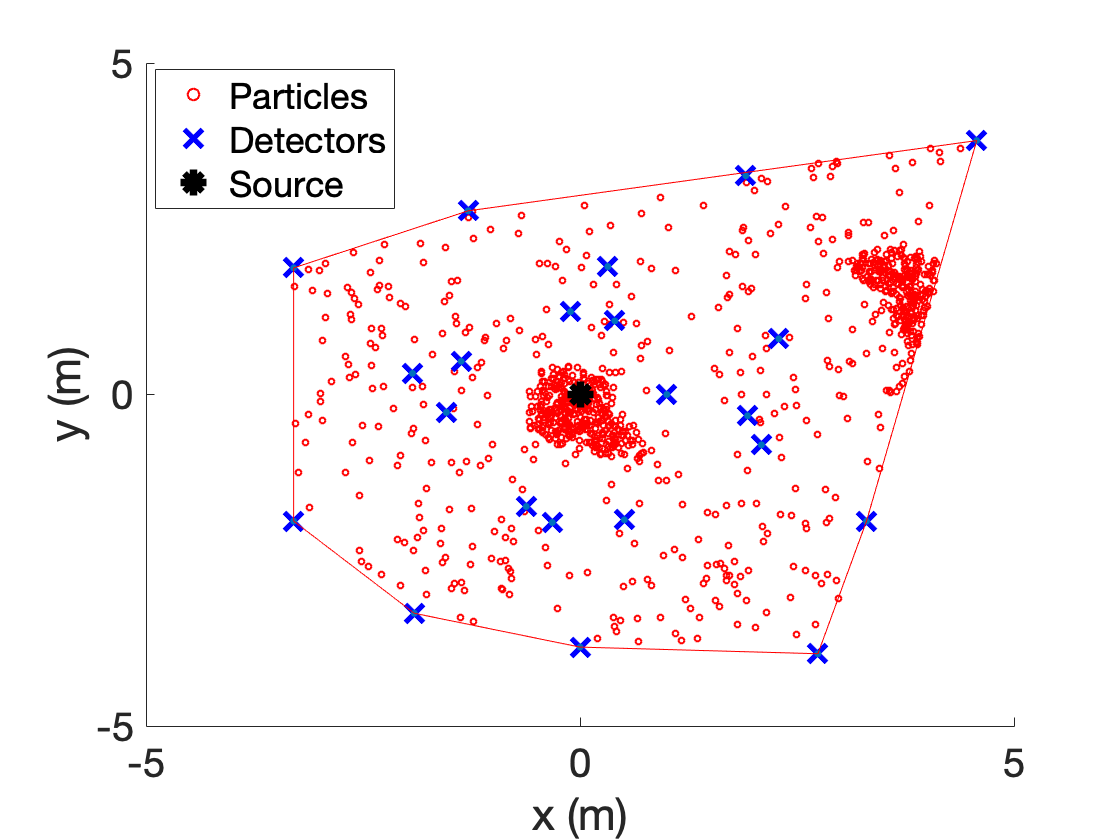}
    \end{subfigure}
    \begin{subfigure}[b]{0.49\textwidth}
        \includegraphics[width=\textwidth]{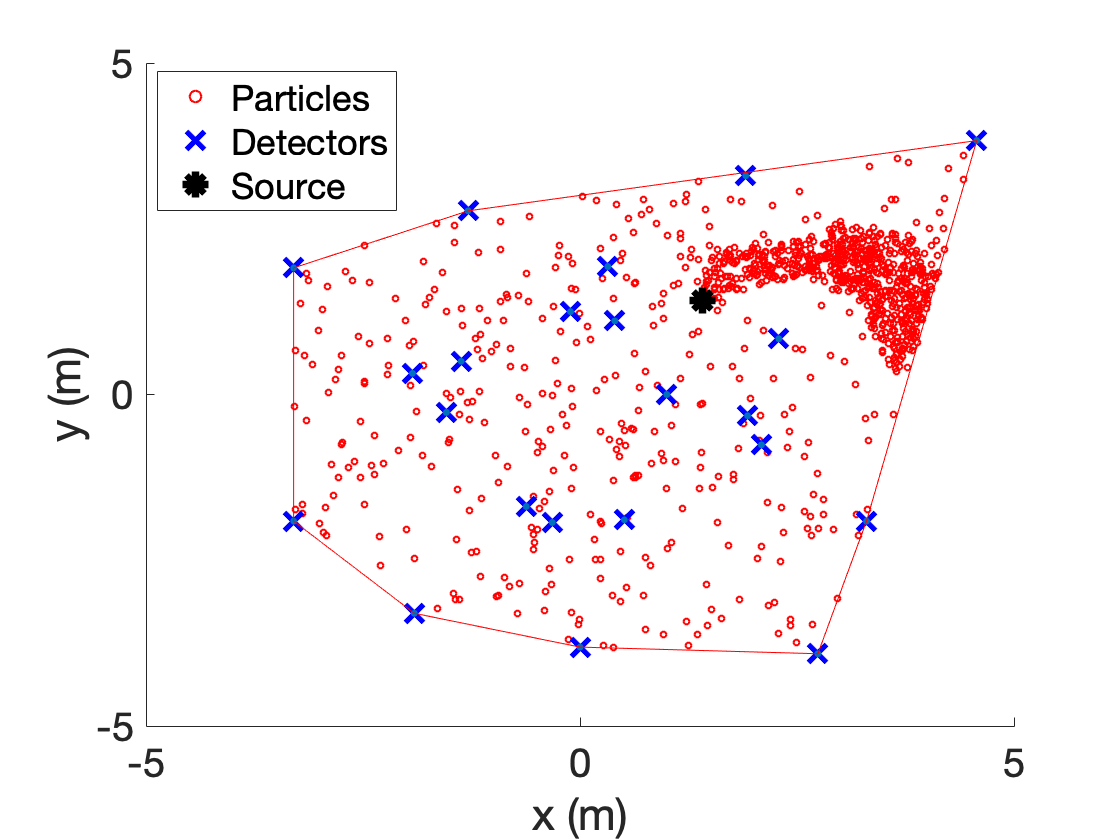}
    \end{subfigure}    
    \caption{Particle filtering results employing a uniform sampling distribution over the convex hull of the detector network for the (a) LSI~C01 and (b) LSI~C03 experimental layouts.}
    \label{fig-conv}
\end{figure}

\begin{figure}[H]
    \centering
    \begin{subfigure}[b]{0.49\textwidth}
        \includegraphics[width=\textwidth]{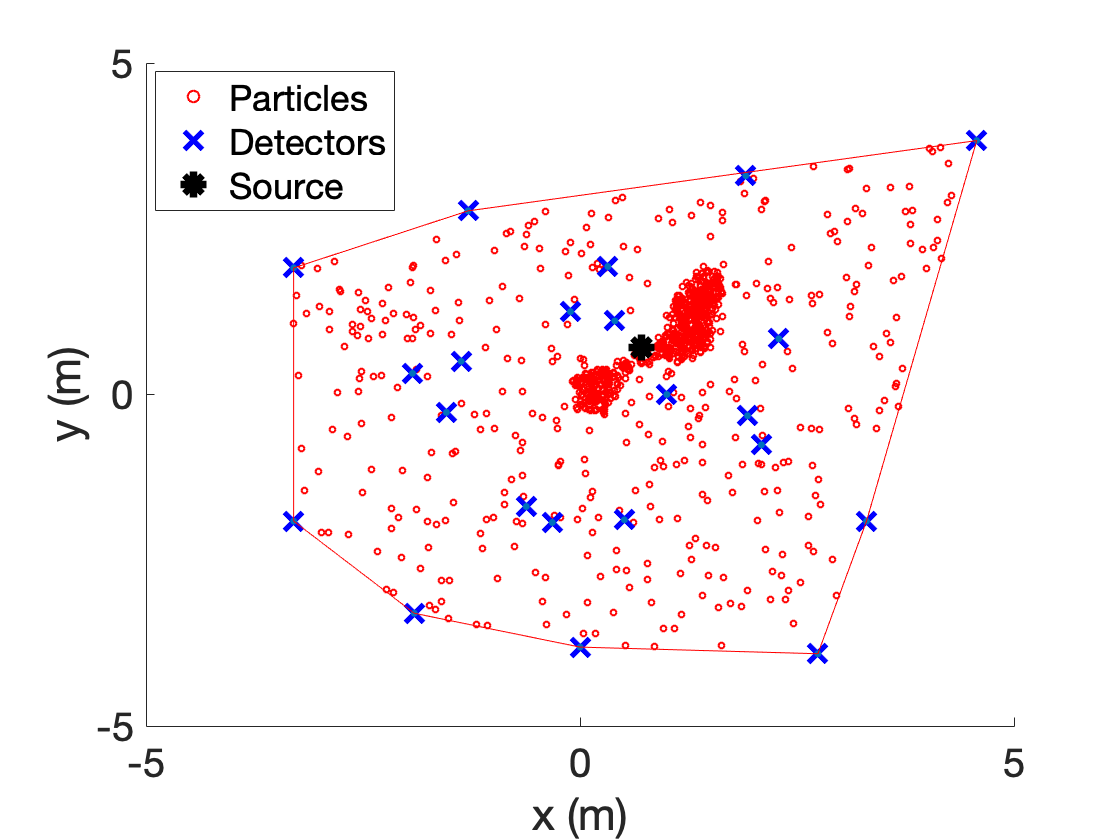}
    \end{subfigure}
    \begin{subfigure}[b]{0.49\textwidth}
        \includegraphics[width=\textwidth]{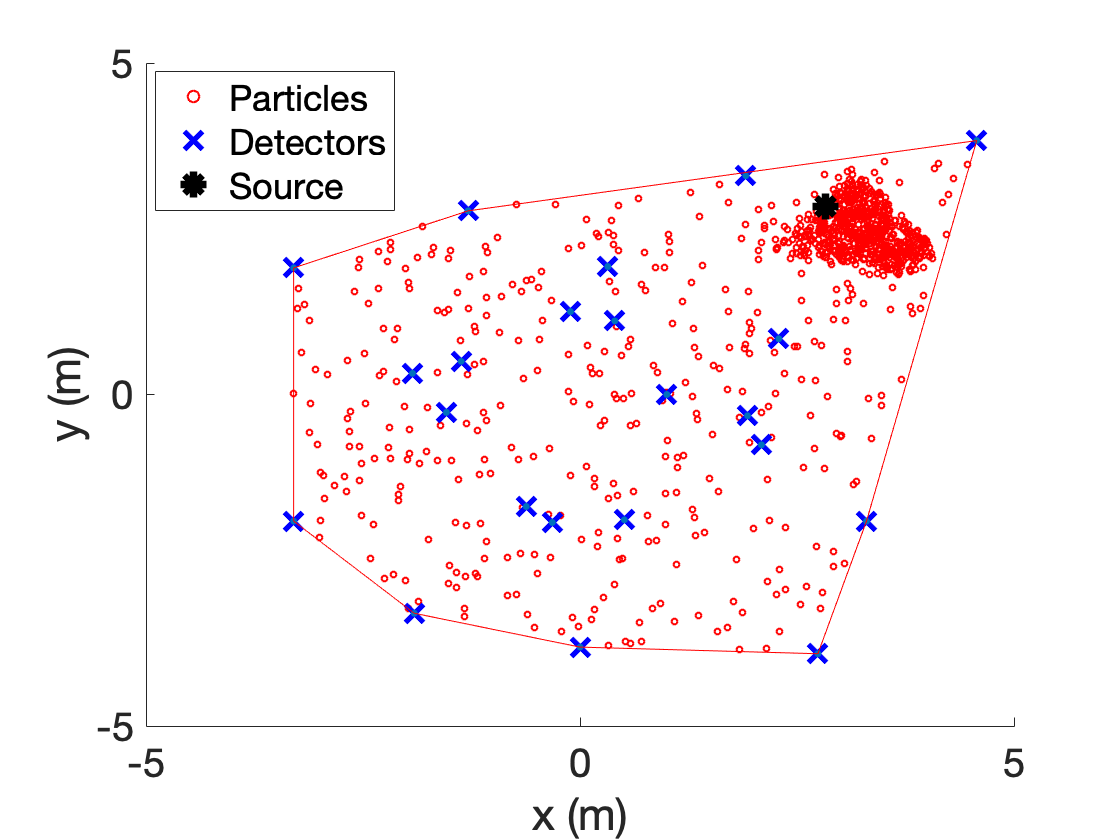}
    \end{subfigure}    
    \caption{Results for the particle filtering algorithm employing a uniform distribution over the convex hull of the detectors, a poisson distribution to simulate additional data, and applied to experimental test cases (a) LSI~C-01 and (b) LSI~C-03.}
    \label{fig-convPois}
\end{figure}


\section{Numerical Results for Simulated Data}\label{raytrace_res}
To assess the performance of this particle filtering algorithm on a realistic urban source localization scenario, we simulate a domain that is a $250 \text{ m} \times 180 \text{ m}$ block in downtown Washington, D.C. \cite{Razvan, Bkgrd}. We construct a 2-D representation of the domain using data from the OpenStreetMaps database.

We treat the buildings as disjoint polygons of uniform density and composition. A satellite photo with the building cross-sections overlaid and detectors plotted as diamond marks is provided in Figure~\ref{fig-Domain}. We semi-randomly assign each building an optical thickness between 1 and 5 mean free paths, with the larger and more dense buildings having larger optical thicknesses, as in \cite{cook2018surrogate}.  The detector locations were randomly selected across the domain 
\begin{equation*}
\Omega = [0,250 \ \text{m}]\times [0,180 \ \text{m}]\times [5\times 10^8,5\times 10^{10} \  \text{Bq}]
\end{equation*}
while excluding the walls of the buildings. We assume the intensity $I$ of a source within this domain is between $5\times 10^8$ and $5\times 10^{10}$ Bq. This corresponds to a source between approximately one hundredth of a curie and one curie. We note that this is multiple orders of magnitude greater than the source strengths we were employing in Section~\ref{num_res}, however we are now considering a larger domain with attenuating materials present.

We assume the detectors have facial areas $A_i=.0058 \ \text{m}^2$ associated with a 3 inch diameter and 3 inch length, intrinsic efficiencies of $\epsilon_i = 62\%$, and dwell times of $\Delta t_i = 5$ second for $i=1,...,10$. These efficiencies are typical for a standard cilyndrical NaI scintillator measuring 662 keV gamma particles. We use a background rate of $B=300$ counts per second, which is typical for this type of detector in an urban environment. We compile the detector locations in Table~\ref{tab-detector_locations}. 

\begin{figure}
\centering
\includegraphics[height=68mm, width=105mm]{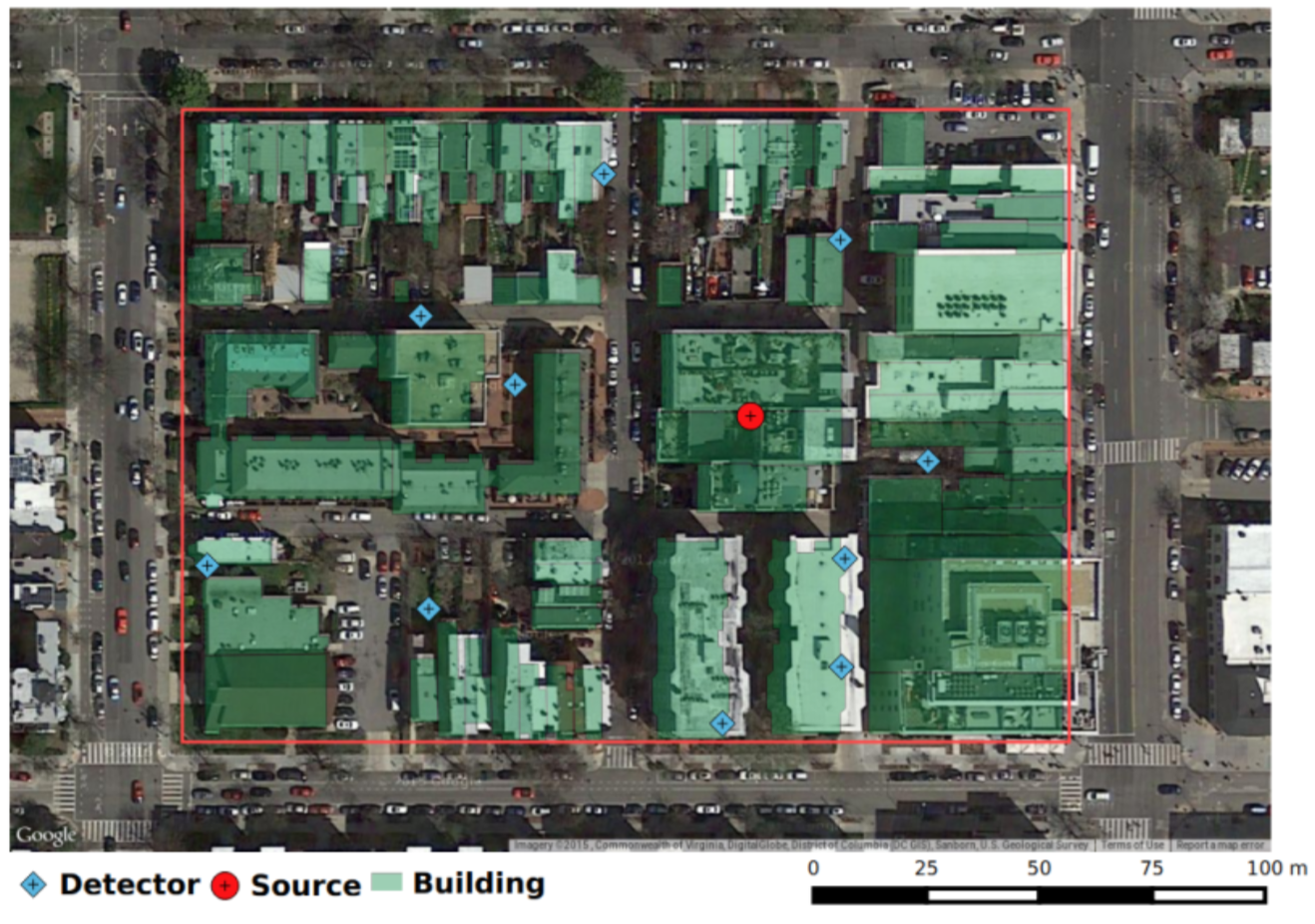}
\caption{Satellite image of domain with overlaid model geometry from \cite{Razvan}.}
\label{fig-Domain}
\end{figure}

\begin{table}[!t]
    \centering
    \caption{Location of NaI detectors plotted in Figure~\ref{fig-Domain}.}
    \begin{tabular}{|c|cc||c|cc|}
    \hline \hline
        Detector & x (m) & y (m)  & Detector & x (m) & y (m)  \\ \hline
        1 & 68.8 & 35.8 & 6 & 189.2 & 19.2 \\
        2 & 66.4 & 119.5 & 7 & 154.5 & 3.0  \\
        3 & 4.1 & 48.1 & 8 & 188.9 & 141.3  \\
        4 & 190.2 & 50.1 & 9 & 119.9 & 160.0  \\
        5 & 94.0 & 99.9 & 10 & 214.5 & 77.9  \\
    \hline \hline
    \end{tabular}
    \label{tab-detector_locations}
\end{table}

Here, we consider two cases in which we simulate a 8.7 mCi source at the locations $(s_0,\ t_0)= (158, \ 98)$ and $(s_0,\ t_0)= (120, \ 40)$, as in \cite{Hite_thesis}. We will call these Case 1 and Case 2 respectively. We employ the ray-tracing model (\ref{eq-numerMod}) to simulate detector responses and employ a simple particle filtering algorithm to perform Bayesian inference. We denote the true source location and intensity as $q_0=(r_0,I_0) = (s_0, t_0, I_0)$. Radioactive decay and detection are Poisson random processes so we take the detector response to be a Poisson-distributed random variable. To simulate detector observations $y_j, \ j=1,\hdots, d$, we sample from the Poisson distribution
\begin{equation}\label{ch-raytrace:eq-init_stat}
y_i \sim P(u^{RT}_i(q_0)+ B_i\Delta t_i), \ \  i=1,...,d.
\end{equation} 
Here, $P(\cdot)$ denotes the Poisson distribution with mean $u^{RT}_i(q_0) + B_i \Delta t_i$, for $u^{RT}_i(q_0)$ given by (\ref{eq-numerMod}). The responses of the $d$ detectors are mutually independent \cite{Cook_thesis}. 

We employ the ray-tracing model in the weight calculation (\ref{eq-weights}) in place of the simple quadratic attenuation detector model employed in Section~\ref{num_res}. We note that while the ray-tracing model is efficient enough for many applications, we must call the model once per particle simulated per time step. Therefore, we are unable to simulate many particles over many time steps. One approach would be to construct accurate and efficient surrogate models to approximate the ray-tracing model \cite{cook2018surrogate, Cook_thesis}. However, we employ the ray-tracing model within the particle filtering framework from Algorithm~\ref{alg-pf} using 100 measurements and 1000 particles. We remind the reader that the detectors we employ in this paper have a five second dwell time. We note also that we sample particles from the convex hull of the detector network as we did in Section~\ref{num_res}; i.e., we employ an importance function in Algorithm~\ref{alg-pf} that is uniform over the convex hull of the detector network. Even with this low number of measurements and particles, employing the ray-tracing model for this problem requires approximately three hours to compute the posterior. 

We simulate the conditions for Case 1 and attempt to localize this source using the detector measurements provided by the ray-tracing model. We plot the particles remaining after resampling in Figure~\ref{fig-pf_gefry1} and we observe that we have localized the source to within approximately 10.5 meters. We next consider Case 2 and plot the particle filtering results in Figure~\ref{fig-pf_gefry1}~(b). We note that we are able to localize this source to within approximately 8 meters. In these figures, we plot the true source location, the mean particle $(x,y)$ location, and the detectors. For many applications, this localization accuracy is sufficient. However, we note that this accuracy can be improved by employing robust resampling strategies which we leave this as future work. 

\begin{figure}[t]
    \centering
    \begin{subfigure}[b]{0.49\textwidth}
        \includegraphics[width=\textwidth]{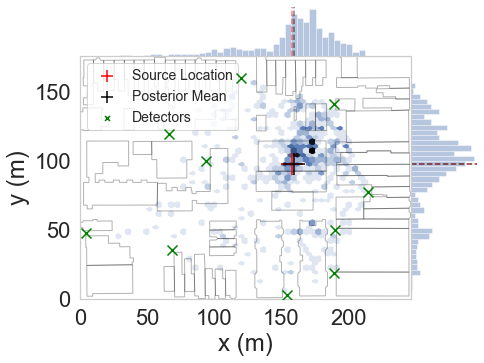}
    \end{subfigure}
    \begin{subfigure}[b]{0.49\textwidth}
        \includegraphics[width=\textwidth]{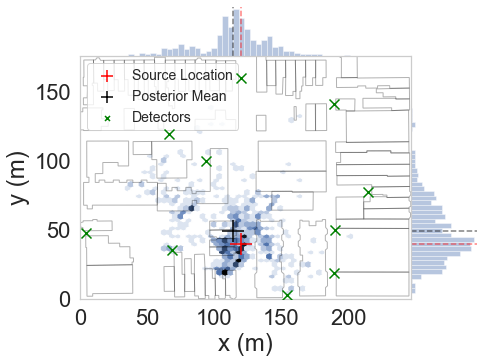}
    \end{subfigure}    
    \caption{Posterior density obtained by employing the ray-tracing model for Bayesian inference with the SIR particle filter from Algorithm~\ref{alg-pf}. We plot the cases where the source is located at (a) $r_S=[158 \text{ m}, 98 \text{ m}]$ and (b) $r_S=[120 \text{ m}, 40 \text{ m}]$.}
    \label{fig-pf_gefry1}
\end{figure}

Prior research into moveable detectors using the ray-tracing model (\ref{eq-numerMod}) and the simulated domain in Figure~\ref{fig-Domain} has employed mutual information with Delayed Rejection Adaptive Metropolis (DRAM), a Markov Chain Monte Carlo (MCMC) algorithm, to determine detector locations that would lower the uncertainty in the posterior distribution \cite{michaud2019simulation}. Additionally, in \cite{Razvan2}, the problem of optimizing detector placement was treated as a combinatorial problem where the sum of squares error in the average sense is minimized by the optimal detector network over a large number of candidate locations. The particle filter detailed in Algorithm~\ref{alg-pf} lends itself to a moving detector strategy, as measurements are taken continuously over time and are employed to inform the particle weights. Because of this, we are able to move the detectors after a certain number of measurements in a way that will decrease the uncertainty in the posterior approximated by the particles.  

To begin, we implement the simple detector movement strategy of moving the detectors each one meter toward the location of the mean of the estimated posterior within the space after each measurement is taken. This is performed by evaluating the mean of the non-resampled particles' positions and moving the detectors one meter toward this location. To avoid detector collisions with buildings, we check if the detector will be moved across the boundary of a building and, if so, we attempt to move it in the coordinate directions toward the source. We first check which coordinate direction would bring it closer to the source estimate and attempt to move the detector in that direction. If a movement in that direction results in a collision with a building, the other coordinate direction is checked. If both of these cases result in a building collision, we randomly select directions in which to move the detector until a direction is found such that a building collision is avoided. In this way, the detectors move toward the source location and do not collide with buildings within the geometry. However, these detectors are not technically moving, but moveable since they do not have a specific trajectory. We plot results for this simple movable detector strategy in Figure~\ref{fig-moving_dets}(a) with the detectors being moved one meter after each measurement is taken. Note that we continue to sample particles from the convex hull of the original detector network. 

We observe in Figure~\ref{fig-moving_dets}(a) that as the detectors are moved through the domain, regions that previously had little probability become candidate locations for the source again. This is especially apparent in the bottom right portion of the convex hull of the original detector network in Figure~\ref{fig-moving_dets}(a). Since the particle weights are reset during each resampling step, the detectors do not retain information on areas from which they have been moved away. In this way, the particles tend toward locations that have already been disqualified as source locations when the detectors are moved farther from those areas. Therefore, we propose that at each detector movement step, we construct a kernel density estimate (KDE) of the approximate posterior and sample particles from the KDE.  We then move the detectors toward the locations of high probability within the domain as before. We note that the KDE tends to smooth the posterior and for this reason should be avoided in some cases \cite{Smith}. However, for this simple problem with a single source, the smoothing of the posterior by the KDE does not negatively affect the results. 

We construct the KDE using \texttt{sklearn}'s \texttt{KernelDensity} \cite{scikit-learn} method after 10 measurements have been taken and sample particles from this distribution as new measurements are taken by the detectors. The detectors are then moved one meter in the direction of the mean of the approximated posterior. This process is repeated every 3 measurements so that the detectors converge toward the posterior's location of maximum probability within the space. We plot the results for this KDE-informed method of sampling in Figure~\ref{fig-moving_dets}(b) and we observe that the posterior is localized around the true source location with a mean value approximately two meters away from the true source location. These results were obtained using 100 measurements from the detectors, each of which has a five second dwell time. We note that this is a simple moveable detector strategy and that future work includes investigation of informed detector movement strategies such as the a Lyapunov-redesign method detailed in \cite{Egorova_moving_sensors}.

\begin{figure}[h]
    \centering
    \begin{subfigure}[b]{0.49\textwidth}
        \includegraphics[width=\textwidth]{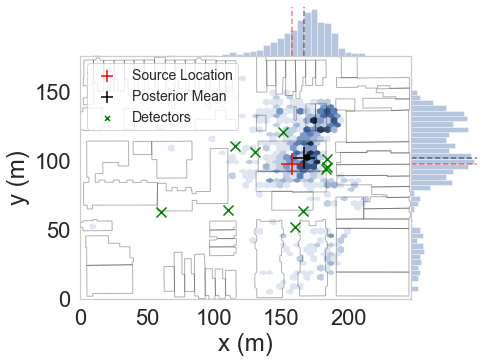}
    \end{subfigure}
    \begin{subfigure}[b]{0.49\textwidth}
        \includegraphics[width=\textwidth]{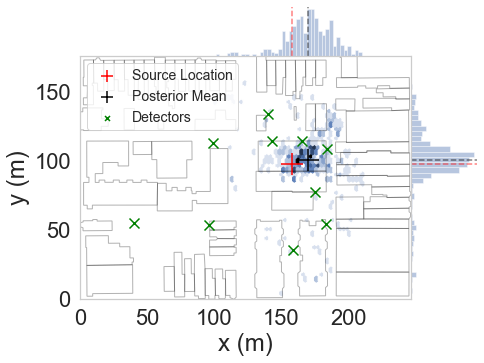}
    \end{subfigure}    
    \caption{Posterior density for Case 1 employing the particle filtering algorithm with moveable detectors. We plot (a) the results when particles are sampled uniformly over the convex hull and (b) the results when particles are sampled from the kernel density estimate of the posterior.}
    \label{fig-moving_dets}
\end{figure}


\section{Convergence Analysis}\label{conv_analysis}

We analyze the convergence of the particle filtering algorithm by dividing the problem into two parts. First, we show that the empirical distribution constructed by the discrete points, which we call particles, converges to the true underlying distribution as the number of particles increases. We next show that the particles, and in effect the underlying distribution, converge to the Dirac distribution centered at the true source location. We do this by showing that the radius of the cluster of particles with the largest weights converges to zero as the number of iterations $k$ goes to infinity. Using both of these results, we determine that the empirical distribution will converge to the Dirac distribution centered at the true source location.

We first consider the convergence of the approximation made by the discrete particles to the underlying distribution; i.e., the solution to the optimal filtering problem \cite{convergence}. We denote the empirical particle filtering distribution at the $k$th time step as $\pi_{k}^N$ and the ``true" distribution, which is the solution to the optimal filtering problem, as $\pi_{k}$. We are concerned with showing that $\mathbb{E}[(\langle\pi_{k}^N,\phi\rangle-\langle\pi_{k},\phi\rangle)]$ is bounded above by some value dependent on the number of particles, $N$, where $\langle \hspace{2pt}\cdot\hspace{2pt},\hspace{2pt}\cdot\hspace{2pt}\rangle$ denotes the Euclidean inner product. We take $\phi\in B(\mathbb{R}^{n_x})$, where $B(\mathbb{R}^{n_x})$ is a Borel $\sigma$-algebra on $\mathbb{R}^{n_x}$ and $n_x$ is the dimension of the state space.

Crisan and Doucet \cite{convergence} outline a theorem that, when applied to this problem, provides the first part of the convergence of the particle filtering algorithm.
\begin{theorem}
Given that the measurement process $h_k(\cdot)$ is a bounded function in argument $x_k\in \mathbb{R}^{n_x}$; i.e., $||h_k||<\infty$ for all $k$, then for any $\phi\in B(\mathbb{R}^{n_x})$, for all $k=1,...,T$, 

\begin{equation*}
\mathbb{E}[(\langle\pi_k^N,\phi\rangle-\langle\pi_k,\phi\rangle)^2]\leq C\frac{||\phi||^2}{N}.
\end{equation*}
\end{theorem}

Next, we show that the particles converge to the Dirac distribution centered at the true source location as more measurements become available. That is, we show that $\mathbb{E}[\langle\pi_k^N,\phi\rangle-\langle\delta_{(x_s,y_s)},\phi\rangle] \rightarrow 0 $ as $k \rightarrow \infty$. We know that for our application, we can obtain this stronger form of convergence for our particle filtering algorithm due to the statistical independence of the measurements and the lack of dependence of our observations on the state variables. We use both of these facts in our analysis to show the particle filter convergence. 

This problem can be reformulated as showing that the ball centered at the true source location and comprised of the highest weighted $f\times 100$ percent of particles at the $(k+1)$st step, has a probability of 1 of being smaller than or equal to the size of the same ball from the $k$th step, as $k \rightarrow \infty$. We denote the radius of these balls as $r_k$ and we show that $\Pr[r_{k+1} \leq r_k] \rightarrow 1$ as $k \rightarrow \infty$. 

For the considered algorithm, the resampling step consists of retaining the top weighted $(1-f)\times100$ percent of the particles and replacing the other $f\times 100$ percent of particles with randomly distributed particles drawn from a uniform distribution across the state space. The numerical results were obtained using a value of $f=0.6$, meaning we resampled $60\%$ of the particles at each time step. Because of this simple resampling strategy, $\Pr(r_{k+1} \leq r_k)$ can be approximated by $\Pr[h\big(p_{k}^f\big) > h(p_k^*)]$. Here, we denote the particles from the $k$th step by $p^i_k$, where $i=1,...,N$. The $(f\times N)$th weighted particle is $p^f_k$ and $p^*_k$ denotes any of the particles that are closer to the source than the particle $p^f_k$. The weight of each particle is given by the weight function $h(p_k^i)$, which is dependent on the Poisson distributed measurements. We note that the distance from $p^f_k$ to the mean of the $(1-f)\times N$ highest weighted particles defines the radius $r_k$. Therefore, the probability that the radius of the ball of the top $f$ percent of the particles is decreasing is equal to the probability that any of the top weights is less than the $(f\times N)$th weight.  

\par Using this relation, we next show that $\Pr[h(p_{k}^f) > h(p_k^*)]\rightarrow 0$ as $k\rightarrow \infty$. 
\begin{theorem}
We assume that the distribution to be approximated is a Dirac distribution centered at the source parameter values and that the measurements are Poisson distributed. We then obtain $\Pr[h(p_{k}^f) > h(p_k^*)]\rightarrow 0$ as $k\rightarrow \infty$.
\end{theorem}
\begin{proof}
We have, 
\begin{align*}
\Pr[h(p_{k}^f) > h(p_k^*)] &= \Pr[h(p_{k}^f)>\epsilon, h(p_k^*)<\epsilon \text{  for any } \epsilon] \\
&= \Pr[h(p_{k}^f))>\epsilon \big|h(p_k^*)<\epsilon]\Pr[h(p_k^*)<\epsilon], \\
\end{align*}
since the measurements are statistically independent. Employing properties of conditional probabilities, we obtain
\begin{align*}
\Pr[h(p_{k}^f))>\epsilon \big|h(p_k^*)<\epsilon]\Pr[h(p_k^*)<\epsilon] &\leq \Pr[h(p_{k}^f))>\epsilon \big| h(p_k^*)<\epsilon] \\ 
&\leq \Pr[h(p_{k}^f))>\epsilon]. 
\end{align*}
We now employ Chebyshev's inequality to obtain
\begin{align*}
 \Pr[h(p_{k}^f))>\epsilon]  &\leq \frac{\mathbb{E}[h(p_{k}^f))]}{\epsilon} \\
&= \frac{\delta_{S}(p_f^k)}{\epsilon} = 0, \text{ as } k \rightarrow \infty,
\end{align*}
where we have additionally used the fact that the distribution we are approximating is a Dirac distribution centered at the true source parameters. 
\end{proof} 

We note that if $\delta_{S}(p^f_k) \neq 0$, then all $f$ percent of the particles are at the true source location. However, if $\delta_{S}(p^f_k) = 0$, then there is a larger than zero probability of the uniformly distributed resampled particles being closer to the true source location than $p^f_k$. Therefore, $P(r_{k+1} \leq r_k) \rightarrow 1$ as $k \rightarrow \infty$.

\section{Conclusions}\label{conclusions}
We have provided numerical and analytic proof of the convergence of this particle filtering algorithm for the problem of determining a source location and intensity from detector observations. The posterior approximation of the source parameters provided by this particle filtering algorithm converges numerically for sources at multiple locations within the domain when we employ a prior constrained to within the convex hull of the detector network. The employed particle filtering algorithm was validated using publicly available experimental data obtained in an open-field environment. We then employed this algorithm to localize a source of radiation within a two-dimensional simulated urban environment. We investigated the use of mobile detectors within the framework of this sequential Monte Carlo method and found that we were able to localize the source to within several meters. Additionally, we proved the convergence of the approximation by the particles to the true underlying distribution and the convergence of the particles to the Dirac distribution centered at the source properties. These convergence results do not necessarily apply to other problems, unless they have the similar properties of statistically independent measurements and a lack of dependence of the measurement process on the state space. 

One direction for future research is to extend this analysis to broader classes of particle filtering methods and source localization problems. This includes analytically investigating the convergence of methods with other resampling algorithms which have already been shown to converge numerically \cite{Camila}. Previous work has been performed to extend this source localization problem to a more complex urban environment employing the MCMC algorithms DRAM and DiffeRential Evolution Adaptive Metropolis (DREAM) algorithm \cite{Schmidt_thesis}. High-fidelity nuclear transport codes have been employed to simulate detector measurements in spatially 3-D urban domains \cite{cook_mcnp}, which have then been used to solve the source localization problem using DRAM. The use of particle filtering algorithms within this framework may facilitate complex moving detector and source strategies, which also comprises future research.

\section*{Acknowledgements}
This research was supported by the Department of Energy National Nuclear Security Administration (NNSA) under the Award Number DE-NA0002576 through the Consortium for Nonproliferation Enabling Capabilities (CNEC).

\newpage

\bibliography{JaredCook}{}
\bibliographystyle{plain}

\end{document}